\documentclass[11pt]{article}
\usepackage[a4paper,margin=1in]{geometry}
\usepackage{amsmath,amssymb,amsthm,mathtools,aliascnt}

\usepackage[disable]{todonotes}

\usepackage[dvipsnames]{xcolor}
\usepackage{enumitem,graphicx,adjustbox,booktabs,tocloft,array,longtable}
\usepackage{microtype}
\usepackage{tikz}
\usetikzlibrary{arrows.meta,positioning,calc,fit}
\usepackage{hyperref}
\usepackage{cleveref}
\hypersetup{colorlinks=true,linkcolor=MidnightBlue,citecolor=MidnightBlue,urlcolor=MidnightBlue,
 pdftitle={Single Exponential FPT Algorithm for Planar Directed Feedback Vertex Set},pdfauthor={Daniel Lokshtanov, Saket Saurabh, and Jie Xue}}
\definecolor{softblue}{RGB}{241,244,248}
\definecolor{softgreen}{RGB}{242,245,242}
\definecolor{softorange}{RGB}{248,246,242}
\definecolor{softgray}{RGB}{247,247,247}
\newtheorem{theorem}{Theorem}[section]
\newaliascnt{lemma}{theorem}
\newtheorem{lemma}[lemma]{Lemma}
\aliascntresetthe{lemma}
\newaliascnt{corollary}{theorem}
\newtheorem{corollary}[corollary]{Corollary}
\aliascntresetthe{corollary}
\newaliascnt{proposition}{theorem}

\aliascntresetthe{proposition}
\newaliascnt{observation}{theorem}
\newtheorem{observation}[observation]{Observation}
\aliascntresetthe{observation}
\newaliascnt{fact}{theorem}

\aliascntresetthe{fact}
\theoremstyle{definition}
\newaliascnt{definition}{theorem}
\newtheorem{definition}[definition]{Definition}
\aliascntresetthe{definition}
\newaliascnt{remark}{theorem}

\aliascntresetthe{remark}
\crefname{theorem}{Theorem}{Theorems}
\Crefname{theorem}{Theorem}{Theorems}
\crefname{lemma}{Lemma}{Lemmas}
\Crefname{lemma}{Lemma}{Lemmas}
\crefname{corollary}{Corollary}{Corollaries}
\Crefname{corollary}{Corollary}{Corollaries}
\crefname{proposition}{Proposition}{Propositions}
\Crefname{proposition}{Proposition}{Propositions}
\crefname{observation}{Observation}{Observations}
\Crefname{observation}{Observation}{Observations}
\crefname{fact}{Fact}{Facts}
\Crefname{fact}{Fact}{Facts}
\crefname{definition}{Definition}{Definitions}
\Crefname{definition}{Definition}{Definitions}
\crefname{remark}{Remark}{Remarks}
\Crefname{remark}{Remark}{Remarks}

\newcommand{\Solve}{\text{\normalfont\scshape Solve}}
\newcommand{\Exact}{\text{\normalfont\scshape Exact}}

\newcommand{\FAIL}{\text{\normalfont\scshape Fail}}
\newcommand{\wt}{w}

\setlist{itemsep=3pt,topsep=5pt}
\usepackage{mdframed,booktabs,fancyhdr}

\newcommand{\OO}{O}

\AtBeginDocument{%
  \let\originaltableofcontents\tableofcontents
  \renewcommand{\tableofcontents}{\par\noindent
    \begin{minipage}{\linewidth}\originaltableofcontents\end{minipage}\par}}

\title{
Single-Exponential Algorithms for Directed Feedback Vertex Set on Planar Digraphs}

\author{%
Daniel Lokshtanov\thanks{%
University of California, Santa Barbara, USA.
Email: \texttt{daniello@ucsb.edu}.}
\and
Saket Saurabh\thanks{%
The Institute of Mathematical Sciences, HBNI, Chennai, India;
and University of Bergen, Norway.
Email: \texttt{saket@imsc.res.in}.}
\and
Jie Xue\thanks{%
New York University Shanghai, China.
Email: \texttt{jiexue@nyu.edu}.}
}
\date{}
\begin{document}
\maketitle

\begin{abstract}
We consider \textsc{Directed Feedback Vertex Set} on planar digraphs,
parameterized by the solution size $k$. We give a randomized algorithm
with one-sided error running in time
$(2+\sqrt5)^k n^{O(1)}= 4.24^k n^{O(1)}$, and a deterministic
algorithm running in time $8.04^k n^{O(1)}$. Both algorithms use
polynomial space. To the best of our knowledge, these are the first
single-exponential fixed-parameter algorithms for \textsc{Directed
Feedback Vertex Set} on planar digraphs. This contrasts with general
digraphs, where the best known algorithms run in time
$2^{O(k\log k)}(n+m)$, and whether a $2^{o(k\log k)}n^{O(1)}$-time
algorithm exists remains a major open problem. Our main tool is an
exact Euler-type counting identity for plane digraphs. 
It shows that every small solution must carry a large share of the vertices whose in- and out-arcs alternate in the embedding, while solutions avoiding such vertices can be computed by reducing to \textsc{Directed Feedback Arc Set}, which is known to be solvable in polynomial time on planar digraphs via the Lucchesi-Younger theorem.
%\todo[inline]{by reduction to FAS on planar digraphs}
\end{abstract}

%\newpage

%\section{Introduction}\label{sec:introduction}
%
% intro_v4.tex -- version 4 (face-parameter result removed)
\section{Introduction}
\label{sec:introduction}
Deleting a small number of vertices to eliminate all cycles
is a fundamental algorithmic problem. In
\textsc{Feedback Vertex Set} (\textsc{FVS}), the input is
an undirected graph $G$ and an integer $k$, and the task is
to find at most $k$ vertices whose deletion makes $G$ acyclic.
Its directed counterpart,
\textsc{Directed Feedback Vertex Set} (\textsc{DFVS}),
asks for at most $k$ vertices that intersect every directed
cycle of a given digraph. Although the two problems have
similar formulations, their parameterized complexity has
developed very differently.

The parameterized study of \textsc{FVS} predates the formal
development of parameterized complexity. Mehlhorn's 1984
monograph already contained a fixed-parameter algorithm
for the undirected problem~\cite{Mehlhorn84,ChitnisEtAl15}.
Since then, \textsc{FVS} has been a central problem in the
development of parameterized algorithms, with techniques
such as randomized branching, iterative compression, and
Cut\& Count contributing to its study 
~\cite{BeckerEtAl00,GuoEtAl06,CyganEtAl11}.
Successive improvements have led to a deterministic
algorithm running in time
$\OO^*(3.460^k)$~\cite{IwataKobayashi19}
and a randomized algorithm running in time
$\OO^*(2.7^k)$~\cite{LiNederlof20}.

% The parameterized study of \textsc{Feedback Vertex Set} predates the
% formal development of parameterized complexity.  Already in 1984,
% Mehlhorn's monograph contained a fixed-parameter algorithm for the
% undirected problem~\cite{Mehlhorn84}; this is commonly credited as the
% first FPT algorithm for \textsc{Feedback Vertex Set}~\cite{ChitnisEtAl15}.
% Since then, undirected \textsc{Feedback Vertex Set} (\textsc{FVS}) has become one
% of the central testbeds for parameterized algorithms.  A long sequence
% of works introduced and refined techniques such as randomized
% branching, iterative compression, and Cut\&Count in the context of this
% problem~\cite{BeckerEtAl00,GuoEtAl06,CyganEtAl11}.  This line of work has
% led to steadily improving dependence on the solution size $k$.  The
% current best deterministic algorithm runs in time
% $\OO^*(3.460^k)$~\cite{IwataKobayashi19}, whereas the current best
% randomized algorithm runs in time $\OO^*(2.7^k)$~\cite{LiNederlof20}.

For \textsc{DFVS}, progress has been slower. Its
fixed-parameter tractability was a major open problem
for nearly two decades, despite algorithms for special
cases such as tournaments~\cite{RamanSaurabh03}.
Chen, Liu, Lu, O'Sullivan, and Razgon settled this question
by giving an algorithm with running time
$\OO(4^k k!\,n^{\OO(1)})$~\cite{ChenEtAl08}.
Lokshtanov, Ramanujan, and Saurabh subsequently obtained
a linear dependence on the input size, with running time
$\OO(k!\,4^k k^5(n+m))$~\cite{LokshtanovRamanujanSaurabh18}.
Xiong and Xiao further improved this bound to
$\OO(k!\,2^{o(k)}(n+m))$~\cite{XiongXiao25}.
These improvements retain a factorial dependence on $k$,
leaving the central question of whether \textsc{DFVS}
admits an algorithm running in time $c^k n^{\OO(1)}$
for some constant $c$.

Planarity is a natural setting in which to seek such an
algorithm. It already helps when \textsc{DFVS} is
parameterized by the treewidth $t$ of the underlying
undirected graph: Bonamy et al.~\cite{BonamyEtAl18}
gave an algorithm running in time
$2^{\OO(t)}n^{\OO(1)}$ on planar digraphs, whereas a
dependence of $2^{o(t\log t)}$ is impossible on general
digraphs under the Exponential Time Hypothesis.
However, this does not yield a single-exponential
algorithm parameterized by the solution size $k$.
Whether such an algorithm exists remained open even
for planar digraphs.

A second long-standing question concerns polynomial preprocessing.
Whether \textsc{DFVS} admits a polynomial kernel parameterized by
the solution size remains open on general
digraphs~\cite{BergougnouxEtAl21}. Until very recently, this was
open even on planar digraphs. Sheng and Xiao recently gave a
deterministic kernel for \textsc{Planar DFVS} with
$\OO(k^{66}\log^2 k)$ vertices and arcs~\cite{ShengXiao26}.

A subtlety is that their kernelization may increase the solution
budget: the output parameter $k'$ is bounded by
$\OO(k^{66}\log^2 k)$, rather than by $k$.
If the kernel preserved the budget, enumerating all subsets
of at most $k$ vertices in the reduced instance would give
running time $2^{\OO(k\log k)}n^{\OO(1)}$.
Here, however, we must allow subsets of size up to $k'$,
so this argument does not apply. Naively enumerating all
vertex subsets gives only
$2^{\OO(k^{66}\log^2 k)}n^{\OO(1)}$ time.
Using the treewidth algorithm of~\cite{BonamyEtAl18}
improves this to $2^{\OO(k^{33}\log k)}n^{\OO(1)}$,
since a planar graph on $N$ vertices has treewidth
$\OO(\sqrt N)$. Thus, although the kernel resolves the
polynomial-kernel question on planar digraphs, it does not
by itself resolve the single-exponential running-time question.

% A second long-standing question concerns polynomial preprocessing.
% Whether \textsc{DFVS} admits a polynomial kernel parameterized by
% the solution size remains open on general
% digraphs~\cite{BergougnouxEtAl21}. Until very recently, this was
% open even on planar digraphs. Sheng and Xiao recently gave a
% deterministic kernel for \textsc{Planar DFVS} with
% $\OO(k^{66}\log^2 k)$ vertices and arcs~\cite{ShengXiao26}.
% Their kernelization may increase the solution budget: the output
% parameter $k'$ has the same bound $\OO(k^{66}\log^2 k)$ as the
% output size, and is not guaranteed to be at most $k$.
% Combined with the treewidth algorithm of~\cite{BonamyEtAl18},
% the kernel gives running time
% $2^{\OO(k^{33}\log k)}n^{\OO(1)}$, since a planar graph on
% $N$ vertices has treewidth $\OO(\sqrt N)$.
% Thus, although this resolves the polynomial-kernel question
% on planar digraphs, the single-exponential running-time
% question remained unresolved.

\paragraph{Our results.}
We resolve the latter question.  Our main result is the first
single-exponential fixed-parameter algorithm for \textsc{DFVS} on planar
digraphs, parameterized by the solution size $k$.

\begin{theorem}[main theorem]
\label{thm:main-intro}
\textsc{Planar Directed Feedback Vertex Set} admits a randomized algorithm with running time
\[
    (2+\sqrt5)^k n^{\OO(1)}
    \leq 4.2361^k n^{\OO(1)},
\]
and a deterministic algorithm with running time
\[
    8.03^k n^{\OO(1)}.
\]
Both algorithms use polynomial space.
\end{theorem}

Our algorithms use a structural relation between the alternation
of incoming and outgoing arcs around vertices of a plane digraph
and its directed facial walks. 
For DFVS we may assume without loss of generality that every vertex has at least ine in-neighbor and at least one out-neighbor. 
A vertex is \emph{bimodal} if
its incoming arcs occur consecutively in the cyclic order
around it, as do its outgoing arcs; otherwise, it is
\emph{nonbimodal}. 

It turns out that the following problem is polynomial time solvable: given a plane digraph $D$, find the smallest feedback vertex set which only contains bimodal vertices. 
This polynomial time algorithm follows from a simple reduction to \textsc{Planar Directed Feedback Arc Set} (\textsc{Planar DFAS}), which in turn known to be solvable in polynomial time~\cite{LY78,MSW19} \todo{!}.
%
%Using the known polynomial-time algorithm for \textsc{Planar DFAS} as a subroutine, we can find a minimum DFVS restricted to bimodal vertices. 
%
Our algorithms therefore focus on selecting the nonbimodal vertices of a solution, after which the remaining bimodal vertices can be found in polynomial time. \todo{completion routine is chat gpt lingo}

\paragraph{Vertex weights and DFVS.}
For the counting argument, we delete all arcs between distinct
strongly connected components and discard components containing
no directed cycle. These operations preserve every directed
cycle. We call the resulting plane digraph the \emph{cyclic core}
and count faces separately in the inherited embedding of each
of its strongly connected components.
If the cyclic order of the arcs around a vertex $v$ has
$2b(v)$ switches between incoming and outgoing arcs, define
its \emph{weight} by    $\wt(v)=b(v)-1$. 
For a vertex set $V$, write
$\wt(V)=\sum_{v\in V}\wt(v)$, and let
$L=\wt(V(D))$ denote the total weight of the cyclic core $D$.
Every vertex of $D$ has both an incoming and an outgoing
arc, so its weight is nonnegative. Moreover, a vertex has
weight zero precisely when it is bimodal.

\todo{cyclic core needs a definition}
\todo{do we need to consider components separately? I guess yes to avoid the $r$}
\todo[inline]{w is a mneumonic for weight, we should call it weight}

% \paragraph{Alternation excess and directed faces.}\todo{cyclic core needs a definition}
% Consider first the cyclic core of a plane digraph, obtained after
% discarding parts that cannot participate in a directed cycle and
% working separately inside its nontrivial strongly connected
% components. \todo{do we need to consider components separately? I guess yes to avoid the $r$}  If the cyclic order around a vertex $v$ has $2b(v)$
% switches between incoming and outgoing arcs, define
% \[
%     \wt(v)=b(v)-1,
%     \qquad
%     L=\sum_{v} \wt(v).
% \]
% \todo[inline]{w is a mneumonic for weight, we should call it weight}
% Every vertex of the cyclic core has at least one incoming and one
% outgoing arc, and hence $\wt(v)\ge 0$.  Moreover, $\wt(v)=0$ precisely when
% $v$ is \emph{bimodal}: its incoming arcs occur consecutively around
% $v$, and so do its outgoing arcs.

Our main structural lemma (Lemma~\ref{lem-weight}) relates the total vertex weight $L$ to the weight of a DFVS.
Essentially, we show that $L<2(|X|+w(X))$ for any DFVS $X$ of $D$.
The proof follows by introducing an intermediate quantity, $\ell_{\rm cir}$, which is the number of \textit{circular faces} of $D$.
We say a face $f$ of $D$ is \textit{circular} if there is a facial boundary walk around $f$ that follows the direction of each arc.
On one hand, by applying Euler's formula and counting arguments, we show that $L < \ell_{\rm cir}$.
On the other hand, by observing that every circular face must intersect the DFVS $X$ and $2b(x)$ is at least the number of circular faces incident to $x$ for each $x \in X$, we show that $\ell_{\rm cir} \leq \sum_{x \in X} 2b(x) \leq 2(|X|+w(X))$.
This implies the inequality $L<2(|X|+w(X))$, which is the common starting point of our randomized and deterministic algorithms.

 \todo{do we really need the auxiliary bipartite graph here -- the only fact we ever use is that every corner is an alternation for precisely one of its vertex or its face}
 \todo[inline]{alternation excess is just weight}
 
% Our basic topological lemma is obtained from an auxiliary bipartite
% vertex--face graph whose edges represent corners of the embedding.  \todo{do we really need the auxiliary bipartite graph here -- the only fact we ever use is that every corner is an alternation for precisely one of its vertex or its face}
% A corner is colored according to whether the two incident arcs switch
% between incoming and outgoing at the vertex, or equivalently whether a
% facial traversal changes its agreement with the arc directions.  Euler's
% formula then gives an exact counting identity.  One consequence is that,
% if the cyclic core has $r$ strongly connected components and
% $F_{\rm dir}$ directed facial boundary walks, then
% \[
%     F_{\rm dir}\ge L+2r.
% \]
% On the other hand, every directed facial boundary walk must meet every
% DFVS $X$.  Charging such a face to an incident in/out corner at a
% vertex of $X$ yields
% \[
%     F_{\rm dir}
%     \le \sum_{v\in X}2b(v)
%     =2|X|+2\wt(X).
% \]
% Combining the two bounds gives the structural inequality
% \begin{equation}
% \label{eq:intro-weight}
%     \boxed{L+2r\le 2|X|+2\wt(X).}
% \end{equation}
% This inequality is the common starting point of all our algorithms.
% Roughly speaking, if the total alternation excess is large, then a
% solution must contain a substantial amount of this excess.
% \todo[inline]{alternation excess is just weight}
\paragraph{Finding the bimodal part of a solution.}
We now explain the polynomial-time algorithm for finding
a minimum DFVS that only contains bimodal vertices.
For each bimodal vertex $v$, split $v$ into $v^-$ and $v^+$, attach all incoming arcs to $v^-$ and
all outgoing arcs to $v^+$, and add the arc $(v^-,v^+)$.
Since the incoming and outgoing arcs each occur consecutively around $v$, this operation preserves planarity. Assign unit cost to the added arcs and prohibitive cost to all original arcs, leaving nonbimodal vertices unchanged.
Deleting a bimodal vertex $v$ then corresponds to deleting its corresponding arc $(v^-,v^+)$.
Thus the problem reduces to weighted \textsc{Planar Directed Feedback Arc Set}, which is polynomial-time solvable~\cite{LY78,MSW19}. This splitting argument builds on the work of Bai, Cao, and Xiao~\cite{BaiEtAl26}, who showed that \textsc{DFVS} is polynomial-time solvable on planar digraphs in which every vertex has in-degree at most one or out-degree at most one.
%Every such vertex is bimodal. Their result was the starting point of our approach. We use the same argument when nonbimodal vertices are also present but cannot be deleted.

\todo[inline]{i really dont like the name completion here; i think it comes from the intended use that we extend a partial solution to a complete one, but the reader does not know this yet. Personally i like the idea of saying that $L=w(G)$ and now the "completion" problem is just the zero weight case.}

\paragraph{Randomized algorithms.}
As a warm-up, we give a randomized algorithm that
samples vertices proportionally to their weights
while $L>3k$, where $k$ is the remaining budget.
After each selection, we delete the sampled vertex,
decrease the budget by one, and recompute the cyclic
core and the weights. Once $L\le3k$, at most $3k$
vertices are nonbimodal, since each has weight at
least one. We enumerate all subsets of these vertices
of size at most $k$. For each choice, we delete the
selected vertices and find a minimum feedback vertex
set contained in the remaining bimodal vertices,
checking whether the combined solution fits within
the budget. Combining sampling with enumeration
gives a running time of
$6.75^k n^{\OO(1)}$.

Our faster algorithm uses the polynomial-time
algorithm at every step. We first find a minimum
feedback vertex set contained in the currently
bimodal vertices. If its size is at most the
remaining budget, we return it together with the
vertices already selected. Otherwise, we sample
a vertex proportionally to its weight, delete it,
decrease the budget by one, and recompute the cyclic
core and the weights. If the budget or the total
weight is zero when sampling would be required,
the trial fails.

For the analysis, let $k$ denote the remaining
budget and fix a feedback vertex set $X$ of size
at most $k$. Let $s$ be the number of currently
nonbimodal vertices in $X$. If $s=0$, the
polynomial-time algorithm finds a solution within
the budget. Otherwise, each of these $s$ vertices
has weight at least one, so $\wt(X)\ge s$.
By the inequality $L<2(|X|+w(X))$, proportional sampling
selects a vertex of $X$ with probability
\[
    \frac{\wt(X)}{L}
    \ge \frac{\wt(X)}{2(k+\wt(X))}
    \ge \frac{s}{2(k+s)}.
\]
After such a selection, the remaining budget is
$k-1$ and at most $s-1$ vertices of the remaining
target solution are nonbimodal. Their number may
decrease further, since recomputing the cyclic
core can make additional vertices bimodal.

The analysis tracks both $k$ and $s$: a smaller
value of $s$ gives a weaker bound on the probability
of selecting a solution vertex, but also means
that fewer successful selections are needed.
Indeed, after at most $s$ such selections, the
remaining target solution consists entirely of
bimodal vertices and can be found in polynomial
time. We solve the resulting two-parameter
recurrence and show that a single trial succeeds
with probability at least
$\frac12(2+\sqrt5)^{-(k+1)}$, where $k$ now denotes
the initial budget. Repeating the trial independently
$\OO((2+\sqrt5)^k)$ times gives constant success
probability and the randomized running time in
Theorem~\ref{thm:main-intro}. The error is one-sided:
every returned solution is a feedback vertex set
of size at most $k$.

\todo[inline]{$j$ and $s$ seem random isnt $k$ better for $j$?}
\todo[inline]{be more descriptive find smallest FVS which is a subset of the bimodal vertices} 
\todo{find a solution of size at most $k$}
\todo[inline]{there are many points here: when s is zero the poly time algorithm finds the optimal solution, when selecting vertices with probability prop to weight, only the s vertices have non zero probability.}
\todo{write $6.75^k$}

\paragraph{Deterministic algorithm.}
Our deterministic algorithm branches on directed facial
walks containing few deletable nonbimodal vertices.
Here, a vertex is \emph{deletable} if it may still be
included in the solution, and \emph{undeletable}
otherwise. Suppose a circular face $f$ is incident to deletable nonbimodal vertices $v_1,\ldots,v_t$.
Every solution either deletes one of these vertices or intersects the face $f$ at a bimodal vertex.
We therefore create $t$ branches: in branch $i$, we delete $v_i$ and make $v_1,\ldots,v_{i-1}$
undeletable. In one additional branch, we make all of $v_1,\ldots,v_t$ undeletable and \emph{mark}
the face, recording that it must be hit at a bimodal vertex.

Marked faces give a lower bound on the number of bimodal vertices in a solution.
A bimodal vertex belongs to at most two directed facial boundary walks, so hitting $m$ marked faces requires at
least $\lceil m/2\rceil$ bimodal vertices.
Consequently, with remaining budget $k$, a solution contains at most $k-\lceil m/2\rceil$ nonbimodal vertices.
We maintain the marked faces and this bound as vertices are deleted and the cyclic core is recomputed.
If $m>2k$, the branch has no feasible solution and is rejected.

To find a face to branch, we use planarity again. Let $q$ be the number of deletable
nonbimodal vertices. Form the bipartite incidence
graph between these vertices and the unmarked
circular faces, retaining one edge for each distinct vertex-face incidence. This graph is simple and planar.
Its edge bound, together with the inequality $\ell_{\rm cir}\ge L$, implies that whenever $q\ge3m$, some unmarked circular face contains at most $4$ deletable nonbimodal vertices. Branching on this face therefore creates at most $5$ branches.
A deletion branch decreases the budget, whereas the marking branch increases the number of faces that must be hit at bimodal vertices.

When $q<3m$, we instead enumerate all subsets of the deletable nonbimodal vertices of size at most $k-\lceil m/2\rceil$.
For each choice, we delete the selected vertices and use the polynomial-time algorithm to find a minimum DFVS that only contains deletable bimodal vertices.
Thus branching applies when there are many nonbimodal candidates relative to the number of marked faces, while enumeration applies when
there are few candidates and the marked faces limit how many may be selected.
An analysis using a potential of the form $c^k\beta^m$ bounds the combined cost by $8.03^k n^{\OO(1)}$.
The branches and enumerated subsets are processed sequentially, so the algorithm uses polynomial space.

%\paragraph{Remark.}
The deterministic running time can be improved by allowing branching on circular faces with more deletable nonbimodal vertices, thereby postponing enumeration.
For example, branching on faces with at most four such vertices when $q\ge3m$, and at most six when $5m/3\le q<3m$, gives a running
time of $7.589^k n^{\OO(1)}$ with a refined analysis.
Adaptively choosing the bound on the number of deletable nonbimodal vertices from $\{4,\ldots,10\}$ further improves the running time
to $7.435^k n^{\OO(1)}$.
Both variants use polynomial space. We leave these improvements out of the present version and focus on the simpler $8.03^k n^{\OO(1)}$-time algorithm.

\paragraph{Organization.}
Section~\ref{sec:prelims} fixes notation and the cyclic core.
Section~\ref{sec:tools} proves the structural lemma and gives
a single-exponential FPT algorithm for the problem parameterized by the number of nonbimodal vertices.
Section~\ref{sec:randomized} presents the two randomized algorithms,
and Section~\ref{sec:deterministic} the deterministic one.
Section~\ref{sec:conclusion} lists open problems.

% prelims_v2.tex -- version 2
% Preliminaries. Defines everything the later sections use without
% definition: loops, protected/permitted vertices, plane embeddings and
% inherited embeddings, bimodality, the cyclic core, and the entropy
% estimates. New labels: sec:prelims, obs:loops, lem:core,
% lem:switch-monotone, fact:entropy.
% New citation keys to add to the bibliography: HT74, Tar72.

\section{Preliminaries}
\label{sec:prelims}

\paragraph{Digraphs and the problem.}
Digraphs may have parallel arcs. For a digraph $D$ we write $V(D)$ and
$E(D)$ for its vertex and arc sets, and $D-X$ for the digraph obtained
by deleting $X\subseteq V(D)$. A \emph{directed feedback vertex set}
(DFVS) of $D$ is a set $X\subseteq V(D)$ such that $D-X$ is acyclic.
In \textsc{Directed Feedback Vertex Set} we are given $D$ and an
integer $k$, and ask for a DFVS of size at most $k$;
\textsc{Planar DFVS} is the restriction to digraphs whose underlying
undirected graph is planar. A DFVS of size at most the current budget
is called a \emph{solution}; when we fix a solution only for the
analysis, we call it a \emph{target}.

%Several of our algorithms maintain a set $P$ of \emph{protected} vertices, which may not be deleted. The remaining vertices are \emph{permitted}, and a solution is \emph{permitted} if it contains only permitted vertices.

\begin{observation}\label{obs:loops}
A vertex carrying a loop belongs to every DFVS. Hence we may delete all
such vertices, reduce the budget by their number, and reject if the
budget becomes negative. From now on all digraphs are loopless.
\end{observation}

\paragraph{Embeddings.}
A \emph{plane digraph} is a digraph together with a crossing-free
drawing of its underlying graph in the plane. The drawing determines,
at every vertex, a cyclic order of its arc incidences, and it determines
the faces and facial walks used in Section~\ref{sec:bipartite}. Given a
planar digraph, a plane embedding can be computed in linear
time~\cite{HT74}; we fix one at the start, and all quantities below
refer to this embedding. Deleting vertices or arcs from a plane digraph
yields a plane digraph whose embedding is the restriction of the
original one; we call it the \emph{inherited} embedding. None of our
results depends on the choice of the initial embedding.
% structural_tools_v2.tex -- version 2
% ---------------------------------------------------------------------
% Merged structural section: old Sections 1-3 become Subsections 2.1-2.3
% (numbering depends on frontmatter). Old section labels are kept on the
% subsections, so \ref{sec:bipartite}, \ref{sec:alternation} and
% \ref{sec:completion} still resolve. New labels: sec:tools,
% cor:circular-faces, eq:euler, eq:circular-face-formula, lem:alt.
% lem:alt was referenced but undefined in the old draft; it is now the
% "weight share" corollary.
% ---------------------------------------------------------------------

\section{Structural tools}
\label{sec:tools}

This section collects the three ingredients used by all our algorithms.
Subsection~\ref{sec:bipartite} proves an exact Euler-type identity for
plane digraphs. Subsection~\ref{sec:alternation} uses it to show that
every small solution carries a large share of the alternation excess.
Subsection~\ref{sec:completion} shows that solutions avoiding all
nonbimodal vertices can be found in polynomial time.

\subsection{Corners and circular faces}
\label{sec:bipartite}

Throughout this subsection, $D$ is a weakly-connected plane digraph with
at least one arc, which possibly contains parallel arcs but does not contain self-loops.
We write $n=|V(D)|$ and $m=|E(D)|$.
Also, let $F(D)$ denote the set of faces (including the outer face) of $D$.
By Euler's formula, we then have $n-m+|F(D)|=2$.
%Euler's formula implies that
%\begin{equation}\label{eq:euler}
  %n-m+|F(D)|=2.
%\end{equation}

%\paragraph{Corners and facial walks.}
Consider a vertex $v \in V(D)$.
The plane embedding of $D$ induces a cyclic order of the $\deg_D(v)$ arcs
incident to $v$.
A \emph{corner} at $v$ is a pair of arcs incident to $v$ that are consecutive in the cyclic order, together with the angular sector between them.
The sector lies in a unique face $f \in F(D)$, and we say the corner {\em lies in} the face $f$.
We say the corner is \textit{alternating} if one of the two arcs is incoming and the other one is outgoing; otherwise, we say the corner is \textit{consistent}.
Note that the number of alternating corners at $v$ is always even.
Define $b(v) \in \mathbb{N}$ as the number such that there are $2 b(v)$ alternating corners at $v$.
If $v$ is a degree-$1$ vertex, then the unique arc incident to $v$ is consecutive to itself and thus $v$ has exactly one corner (which is consistent).
By our construction, every vertex $v \in V(D)$ has $\deg_D(v)$ corners, and $D$ has $2m$ corners in total.

\begin{lemma}
    For a face $f \in F(D)$, the number of consistent corners lying in $f$ is even.
\end{lemma}
\begin{proof}
Consider a facial walk around $f$ that traverses the boundary of $f$.
The pairs of consecutive arcs on this walk one-to-one correspond to the corners lying in $f$.
Each step of the walk can be either a \textit{forward move} (which follows the direction of the arc) or a \textit{backward move} (which opposes the direction of the arc).
Whenever the walk visits a consistent corner, it switches between forward moves and backward moves.
Note that the number of such switches has to be even.
As a result, the number of consistent corners lying in $f$ is even.
%A vertex that occurs several times on the walk contributes one corner per occurrence, and a bridge is traversed twice.
%All counts in this subsection are counts of corners, that is, of occurrences, and not of distinct vertex-face pairs.
\end{proof}

By the above lemma, for each face $f \in F(D)$, there exists a number $a(f) \in \mathbb{N}$ such that exactly $2 a(f)$ consistent corners lie in $f$.
We say $f$ is a \textit{circular face} if $a(f) = 0$, or equivalently, there is a facial
walk of $f$ (in one of its two traversal directions) that follows the directions of the arcs.
The total number of consistent corners is $\sum_{f \in F(D)} 2a(f)$ and the total number of alternating corners is $\sum_{v \in V(D)} 2b(f)$.
Since the total number of corners is $2m$, we have $\sum_{f \in F(D)} a(f) + \sum_{v \in V(D)} b(f) = m$.
Based on this, we further establish a formula for the number of circular faces.
\begin{lemma} \label{lem-diface}
    Let $\ell_{\rm dir}(D)$ be the number of circular faces of $D$.
    Then
\begin{equation}\label{eq:directed-face-formula}
  \ell_{\rm dir}(D)
  =2+\sum_{\substack{f\in F(D)\\ a(f)\ge1}}\bigl(a(f)-1\bigr)+\sum_{v\in V(D)}\bigl(b(v)-1\bigr).
\end{equation}
\end{lemma}
\begin{proof}
By Euler's formula, it holds that $n-m+|F(D)| = 2$.
Combining this with the equality 
$$\sum_{f \in F(D)} a(f) + \sum_{v \in V(D)} b(f) = m$$ 
we have 
$$\sum_{f \in F(D)} (a(f)-1) + \sum_{v \in V(D)} (b(f)-1) = -2\mbox{,}$$
Write $F'(D) = \{f \in F(D): a(f) \geq 1\}$.
Note that 
$$\sum_{f \in F(D)} (a(f)-1) = \sum_{f \in F'(D)} (a(f)-1) - \ell_{\rm dir}(D)\mbox{.}$$
Therefore, $\ell_{\rm dir}(D) = 2+\sum_{f \in F'(D)} (a(f)-1) + \sum_{v \in V(D)} (b(f)-1)$.
\end{proof}

\subsection{Weight of a DFVS}
\label{sec:alternation}

For each vertex $v \in V(D)$, we define $w(v) = b(v)-1$ as the \textit{weight} of $v$.
Write $w(V) = \sum_{v \in V} w(v)$ for a subset $V \subseteq V(D)$, and $L = w(V(D))$.
We call a vertex $v \in V(D)$ \textit{bimodal} if $w(v) \leq 0$, and \textit{positive} if $w(v) > 0$.
Intuitively, a vertex $v$ is bimodal if the incoming arcs (resp., outgoing arcs) at $v$ are consecutive in the cyclic order around $v$.

\begin{lemma} \label{lem-weight}
%Assume every vertex $v \in V(D)$ has in-degree and out-degree at least $1$ in $D$.
For a directed feedback vertex set $X$ of $D$, we have
\begin{equation}
  |X|+ \wt(X) \geq \frac{L}{2} + 1.
\end{equation}
\end{lemma}
\begin{proof}
By Lemma~\ref{lem-diface}, $\ell_{\rm dir}(D) \geq L+2$.
So it suffices to show that $2(|X|+\wt(X)) \geq \ell_{\rm dir}(D)$.
Note that $|X|+\wt(X) = \sum_{x \in X} b(x)$.
For a circular face $f \in F(D)$, its boundary form a directed cycle and hence $f$ must be incident to some $x_f \in X$.
We then charge each circular face $f \in F(D)$ to the corner at $x_f$ that lies in $f$ (which is alternating because $f$ is circular).
Note that each alternating corner at a vertex $x \in X$ gets charged at most once, because a corner can only lie in one face.
It then follows that the number of circular faces is at most $\sum_{x \in X} 2b(x)$, i.e., $\ell_{\rm dir}(D) \leq \sum_{x \in X} 2b(x) = 2(|X|+\wt(X))$.
\end{proof}

\begin{corollary} \label{cor:sample}
Let $X$ be a directed feedback vertex set of $D$ of size at most $k$.
Suppose we randomly sample a vertex $v^*$ of $D$, where each vertex $v \in V(D)$ is sampled with probability $\frac{\wt(v)}{L}$.
If $L > 3k$, then we have $\Pr[v^* \in X] > \frac{1}{6}$.
\end{corollary}
\begin{proof}
By Lemma~\ref{lem-weight}, $\wt(X) > \frac{L}{2} - |X|$.
If $L > 3k$, then $\frac{L}{2} - |X| \geq \frac{k}{2}$, which implies $\wt(X) > \frac{k}{2}$.
Therefore, $\Pr[v^* \in X] = \frac{\wt(X)}{L} \geq \frac{1}{6}$.
\end{proof}

\subsection{Problem with a few positive vertices}
\label{sec:completion}

Recall that $w(v) = b(v)-1$ for $v \in V(D)$.
Let $p$ be the number of positive vertices of $D$, i.e., vertices $v \in V(D)$ satisfying $w(v) > 0$.
In this section, we show that one can solve \textsc{Directed Feedback Vertex Set} on $D$ in roughly $2^p \cdot n^{O(1)}$ time.

%We recall the following classical algorithmic consequence of the weighted Lucchesi--Younger theorem and planar duality~\cite{LY78}.
%A statement of the weighted theorem, and of the usual splitting argument for strongly planar digraphs, can be found in~\cite[Theorems~4.10 and~4.12]{MSW19}.

We recall the \textsc{(Weighted) Feedback Arc Set} problem, which aims to compute, for a given edge-weighted digraph $G$, a minimum-weighted feedback arc set (FAS) of $G$ (here a set $E$ of edges is a \textit{feedback arc set} if $G-E$ is acyclic).
The problem is known to be polynomial-time solvable if $G$ is planar.

\begin{theorem}[Theorems~4.10 and~4.12 of \cite{MSW19}]\label{thm:fas}
\textsc{Weighted Feedback Arc Set} on planar digraphs admits a polynomial-time algorithm.
%Given a planar digraph with nonnegative integer arc costs, a minimum-cost set of arcs meeting every directed cycle can be found in time polynomial in the input size, including the bit length of the costs.
\end{theorem}

To describe our DFVS algorithm, we first consider a variant of the problem in which all positive vertices are ``protected'' in the sense that they cannot be included in the solution, and we show that this variant can be solved in polynomial time.
Formally, for a subset $P \subseteq V(D)$, we say a DFVS $S$ of $D$ is \textit{$P$-avoiding} if $S \cap P = \emptyset$.

\begin{lemma}\label{lem:protected}
Let $D$ be a plane digraph and $P \subseteq V(D)$ be a subset that contains all positive vertices of $D$.
Then one can compute in polynomial time a minimum $P$-avoiding directed feedback vertex set of $D$.
\end{lemma}
\begin{proof}
First observe that without loss of generality, we can assume that $b(v) \geq 1$ for all $v \in V(D)$.
Indeed, as long as there is a vertex $v \in V(D)$ with $b(v) = 0$, we can simply remove $v$ and solve the problem instance on $(D - v, P \backslash \{v\})$, because $v$ is either a source or a sink in $D$ and thus no directed cycle of $D$ contains $v$.
As such, we can keep removing vertices $v$ with $b(v) = 0$ from $D$.
This ends up with an instance in which all vertices $v$ satisfy $b(v) \geq 1$, or equivalently, $w(v) \geq 0$.

Let $Q = V(D) \backslash P$.
Since $P$ contains all positive vertices of $D$, we have $w(v) = 0$ for all $v \in Q$.
Thus, the vertices in $Q$ are bimodal.
Consider a vertex $v \in Q$.
%We have $w(v) = 0$ and thus $b(v) = 1$.
%This implies that there are exactly two alternating corners at $v$.
%In other words, the incoming arcs (resp., outgoing arcs) at $v$ are consecutive in the cyclic order of the arcs incident to $v$.
We can split $v$ into two vertices $v^-$ and $v^+$ as follows.
We replace $v$ with $v^-$ and $v^+$ connected with an arc $e_v = (v^-,v^+)$.
Each incoming arc $(u,v)$ is then replaced with $(u,v^-)$, and each outgoing arc $(v,u)$ is then replaced with $(v^+,u)$.
As $v$ is bimodal, the incoming arcs (resp., outgoing arcs) at $v$ are consecutive in the cyclic order, and thus this split operation preserves the planarity of the digraph.
See Figure~\ref{fig:split} for an illustration.
We apply the split operation to every vertex in $Q$, and let $D^*$ be the resulting planar digraph.

\begin{figure}[ht]
\centering
\begin{tikzpicture}[>=Stealth,vertex/.style={circle,fill=black,inner sep=1.6pt},scale=.95]
\begin{scope}[xshift=0cm]
\node[vertex,label=below:$v$] (v) at (0,0) {};
\foreach \y in {-.8,0,.8}{\draw[->] (-1.5,\y)--(v);\draw[->] (v)--(1.5,\y);}
\node at (0,-1.25) {before};
\end{scope}
\node at (2.6,0) {$\Longrightarrow$};
\begin{scope}[xshift=5.4cm]
\node[vertex,label=below:$v^-$] (a) at (-.45,0) {};
\node[vertex,label=below:$v^+$] (b) at (.45,0) {};
\foreach \y in {-.8,0,.8}{\draw[->] (-1.9,\y)--(a);\draw[->] (b)--(1.9,\y);}
\draw[->,thick] (a)--node[above]{$e_v$}(b);
\node at (0,-1.25) {after};
\end{scope}
\end{tikzpicture}
\caption{Splitting a vertex in $Q$ into two while preserving the planarity.}
\label{fig:split}
\end{figure}
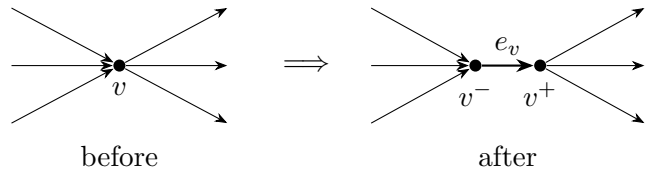

We claim that a set $S \subseteq Q$ is a DFVS of $D$ iff $A_S = \{e_v \in E(D^*): v \in S\}$ is an FAS of $D^*$.
To see this, observe that every (directed) cycle in $D$ one-to-one corresponds to a (directed) cycle in $D^*$ by replacing each vertex $v \in Q$ on the cycle with the two vertices $v^-,v^+$ connected by the arc $e_v$.
One can easily verify that this correspondence is one-to-one and satisfies the following: a cycle in $D$ contains a vertex $v \in Q$ $v \in Q$ iff its corresponding cycle in $D^*$ contains the arc $e_v$.
%\todo{are directed minors defined? Cycles obviously preserved? Better: every directed cycle in $D$ corresponds to a directed cycle in $D'$ by replacing each vertex $v$ by $v-$, $v+$, this correspondence is a bijection}
%As such, a (directed) cycle in $D^*$ uniquely corresponds to a (directed) cycle in $D$.
%One can easily verify that this correspondence is one-to-one and satisfies the following condition: a cycle in $D^*$ contains an arc $e_v$ for $v \in Q$ iff its corresponding cycle in $D$ contains the vertex $v$.
This implies that no cycle in $D$ is disjoint from $S$ iff no cycle in $D^*$ is disjoint from $A_S$.

The above observation allows us to reduce the problem of computing a minimum $P$-avoiding DFVS of $D$ to the problem of computing a minimum FAS $A \subseteq \{e_v \in E(D^*): v \in Q\}$ of $D^*$.
Indeed, if we can compute $A$, then the above observation implies that the set $S = \{v \in Q: e_v \in A\}$ is a minimum DFVS of $D$ that is contained in $Q$, i.e., a minimum $P$-avoiding DFVS of $D$.
Note that the latter problem can be further reduced to an instance of \textsc{Weighted Feedback Arc Set} by setting the weights of the edges in $\{e_v \in E(D^*): v \in Q\}$ equal to $1$ and the weights of the other edges equal to $\infty$.
Since $D^*$ is a planar digraph, the minimum FAS $A$ can be computed in $n^{O(1)}$ time by Theorem~\ref{thm:fas}.
This completes the proof of the lemma.
\end{proof}

Based on the above lemma, we can now design our DFVS algorithm parameterized by the number $p$ of positive vertices.
Suppose we want to compute a DFVS of $D$ with size at most $k$.
Let $V^+ = \{v \in V(D): w(v) > 0\}$ consist of the positive vertices of $D$.
We first guess the subset $R \subseteq V^+$ of vertices that are in the desired DFVS.
Since the DFVS we look for has size at most $k$ and $|V^+| = p$, the total number of guesses is $\sum_{h=0}^{\min\{p,k\}} \binom{p}{h}$.
For each guess $R \subseteq V^+$, we apply Lemma~\ref{lem:protected} to compute a minimum $P$-avoiding DFVS $R'$ of $D-R$ for $P = V^+ \backslash R$.
Clearly, $R \cup R'$ is a DFVS of $D$.
Note that if there exists a DFVS $S$ of $D$ with $|S| \leq k$ and $S \cap V^+ = R$, then $S$ is $P$-avoiding and hence $|R \cup R'| \leq |S| \leq k$.
Therefore, as long as the guess is correct, we find a DFVS of $D$ with size at most $k$.
Since the algorithm of Lemma~\ref{lem:protected} takes $n^{O(1)}$ time, the total running time of our algorithm is $\sum_{h=0}^{\min\{p,k\}} \binom{p}{h} \cdot n^{O(1)} \leq 2^p \cdot n^{O(1)}$.

\begin{lemma}\label{lem:exact}
Let $D$ be a plane digraph with $p$ positive vertices.
There exists a deterministic algorithm that can find a directed feedback vertex set of $D$ with size at most $k$ (if it exists) in $\sum_{h=0}^{\min\{p,k\}} \binom{p}{h} \cdot n^{O(1)}$ time and polynomial space.
\end{lemma}

% ---------------------------------------------------------------------
% Randomized algorithms: replaces the old sections
%   "8^k randomized algorithm"  and
%   "A Faster Algorithm: Sampling with polynomial-time completion".
% Uses the labels of structural_tools.tex (lem:weight, lem:alt,
% lem:protected, lem:exact). Kept labels: sec:baseline, thm:baseline,
% lem:success, lem:time, sec:sampling-completion,
% lem:two-parameter-sample, lem:sampling-monotonicity, lem:R-vs-Q,
% lem:uniform. New: sec:randomized, lem:approx-sampling, thm:main,
% eq:rec-P, eq:rec-R, eq:rec-Q, lem:trial.
% Macros assumed from the preamble: \wt, \Oh, \Exact, \FAIL, \Solve.
% ---------------------------------------------------------------------
\section{Randomized algorithms}
\label{sec:randomized}

In this section, we give two randomized algorithms for \textsc{Planar DFVS}.
Our algorithms have one-sided error: they return either a DFVS (with size at most $k$) or $\FAIL$, and may return $\FAIL$ even when a solution exists.
%Both repeatedly sample a vertex with probability proportional to vertex weights $w$ and delete it. By Corollary~\ref{cor:sample}, each sample hits a solution vertex with a probability that we can bound from below. The first algorithm switches to enumeration when the alternation excess is small, giving running time $\OO(6.75^k n^{\OO(1)})$. The second instead tries, before each sample, to find a solution consisting only of bimodal vertices by reducing to polynomial-time solvable \textsc{Planar DFAS} (Lemma~\ref{lem:exact}). This improves the running time to $(2+\sqrt5)^k n^{\OO(1)}$.

Let $D$ be the input planar digraph and $k$ be the parameter.
Our algorithms first delete all vertices carrying self-loops and reduce the budget $k$ accordingly; if $k$ becomes negative, we return $\FAIL$.
As such, we can assume $D$ does not contain a self-loop.
%The empty set is a valid solution when the digraph is acyclic and is distinct from $\FAIL$.
%Both algorithms use the following sampling procedure, which approximates proportional sampling using a bounded number of random bits.
A common step used in both algorithms is the following (weighted) random sample procedure.
Let $\rho: V(D) \rightarrow \mathbb{R}_{\geq 0}$ be a function and $R = \sum_{v \in V(D)} \rho(v)$.
Consider a randomly sampled vertex $v^* \in V(D)$, where the sampling probability of each $v \in V(D)$ is equal to $\frac{\rho(v)}{R}$.
In this case, we say $v^*$ is sampled from $V(D)$ with respect to $\rho$.

\subsection{A simple \texorpdfstring{$6.75^k$}{6.75^k}-time algorithm}
\label{sec:baseline}
Our algorithm is recursive, and roughly works as follows.
In each call, let $k$ be the current budget and $L = w(V(D))$ as before.
We distinguish two cases:
\begin{itemize}
    \item If $L\le3k$, $D$ has at most $3k$ positive vertices. In this case, we solve the instance by the algorithm of Lemma~\ref{lem:exact}.
    
    \item If $L>3k$, we randomly sample a vertex $v \in V(D)$ with respect to $w$.
    By Corollary~\ref{cor:sample}, $v$ is contained in an optimal solution with constant probability.
    We then include $v$ in the solution and recursively solve the problem on $D - v$ with budget $k-1$.
    %we make six independent attempts. Each attempt samples a vertex proportionally to its alternation weight, deletes it, and recursively solves the resulting instance with budget $j-1$. We return a solution as soon as an attempt succeeds.
\end{itemize}
%We will show that, whenever a solution within the remaining budget exists, each sample hits a vertex of that solution with probability at least $1/6$, including the loss from approximate sampling. This allows us to bound the success probability of the recursive procedure.

%Let $k$ be the budget after deleting vertices carrying loops; it remains fixed throughout the recursion. We initially call $\Solve(D,k)$.

Formally, the core of our algorithm is the procedure $\Solve(D,k)$ presented below.
Here $\Exact$ denotes the algorithm of Lemma~\ref{lem:exact}.

\begin{mdframed}[linewidth=.5pt,innertopmargin=8pt,innerbottommargin=8pt]
\noindent\textbf{Procedure $\Solve(D,k)$.}
\begin{enumerate}[label=\arabic*.]
\item \textbf{Preprocessing.}
If $D$ is acyclic, return $\emptyset$.
If $k=0$, return $\FAIL$.
Remove from $D$ all vertices/arcs that do not belong to any directed cycle.

\item Compute the weight $\wt(v) = b(v)-1$ for each $v \in V(D)$ and $L = w(V(D))$.
%their sum $L$, and the set $B$ of nonbimodal vertices.

\item If $L\le3k$, return $\Exact(D,k)$.

\item Make $6$ independent attempts. In each attempt:
\begin{enumerate}[label=(\alph*)]
    \item Sample a vertex $v \in V(D)$ with respect to $w$.
    %using Lemma~\ref{lem:approx-sampling} with weights $\wt$ and $N=4k$.
    \item Recursively call $\Solve(D-v,k-1)$.
    \item If the recursive call returns $Z \subseteq V(D-v)$, then return $Z\cup\{v\}$.
\end{enumerate}

\item Return $\FAIL$.
\end{enumerate}
\end{mdframed}

\medskip

%We now analyze $\Solve$ and show that repeating it independently $\OO(k+1)$ times yields the following result.

\begin{theorem}\label{thm:baseline}
There is a randomized algorithm for \textsc{Planar DFVS}
running in time $\OO(6.75^k n^{\OO(1)})$ and polynomial space.
It always returns either a DFVS of size at most $k$ or
$\FAIL$, and returns a solution with probability at least
$\frac{1}{2}$ whenever one exists.
\end{theorem}
\begin{proof}
We analyze the recursive procedure $\Solve(D,k)$.
%, which enumerates when $L\le3j$ and otherwise makes six independent sampling attempts. The initial budget $k$ remains fixed throughout the recursion.
Clearly, if $D$ does not have a DFVS of size at most $k$, then $\Solve(D,k)$ returns $\FAIL$.
Assume $D$ has a DFVS $X$ with $|X| \leq k$.
By construction, if $\Solve(D,k)$ returns a subset $X'$ of $V(D)$, then $X'$ is a DFVS of $D$ with size at most $k$.
So it suffices to bound the success probability of $\Solve(D,k)$, i.e., the probability that $\Solve(D,k)$ does not return $\FAIL$.
We shall show that the success probability of $\Solve(D,k)$ is at least $\frac{1}{k+1}$.
By induction, we can assume that for every $v \in X$, the success probability of $\Solve(D-v,k-1)$ is at least $\frac{1}{k}$.

In the procedure $\Solve(D,k)$, if $L \leq 3k$, we call $\Exact(D,k)$ and the success probability of $\Solve(D,k)$ is $1$ in this case.
Otherwise, we make $6$ independent attempts.
In each attempt, we sample a vertex $v \in V(D)$ with respect to $w$.
By Corollary~\ref{cor:sample}, $\Pr[v \in X] \geq \frac{1}{6}$.
Furthermore, if $v \in X$, then $\Solve(D-v,k-1)$ is successful with probability at least $\frac{1}{k}$.
As such, each attempt is successful with probability at least $\frac{1}{6k}$.
The overall success probability is then 
\[
    1-\left(1-\frac1{6k}\right)^6
    \ge \frac1k-\frac5{12k^2}
    \ge \frac1{k+1}.
\]

Based on the above discussion, making $\Theta(k)$ independent calls of $\Solve(D,k)$ gives us a success probability at least $\frac{1}{2}$.
To analyze the running time, observe that in every call of the sub-routine \textsc{Exact}, the number of positive vertices is at most $3k$.
As such, by Lemma~\ref{lem:exact}, the time cost of each call is bounded by
\[
\begin{aligned}
    \sum_{i=0}^{k}\binom{3k}{i}\,n^{\OO(1)}
    &\le
    2^k\left(1+\frac12\right)^{3k}n^{\OO(1)}\\
    &=6.75^k n^{\OO(1)}.
\end{aligned}
\]
Hence the worst-case running time of $\Solve(D,k)$ satisfies the recurrence
\[
    T(k)\le
    \max\left\{
        6.75^k n^{\OO(1)},
        \;6T(k-1)+n^{\OO(1)}
    \right\},
    \qquad T(0)=n^{\OO(1)}.
\]
which solves to $T(k)=6.75^k n^{\OO(1)}$.
Finally, $\Theta(k)$ calls of $\Solve(D,k)$ results in an extra $k$ factor in the running time.
Since $k \leq n$, the overall time complexity is still $6.75^k n^{\OO(1)}$.
\end{proof}

\subsection{A \texorpdfstring{$(2+\sqrt5)^k$}{(2+sqrt5)^k}-time algorithm}
\label{sec:sampling-completion}

We improve the previous algorithm by trying, before every
sampling step, to find a solution that contains only the current
bimodal vertices. By Lemma~\ref{lem:exact}, this can be done in polynomial time through a reduction to
\textsc{Planar DFAS}. If such a solution exists within the
remaining budget, we return it together with the vertices
already selected. 
Otherwise, we do the sampling step as usual.
The key is that we only need to sample the nonbimodal part
of a solution. Moreover, deletions may turn further solution
vertices bimodal, reducing the number of successful samples
needed. We capture this by tracking both the remaining
budget $j$ and the number $s$ of nonbimodal vertices in a
fixed solution. Each trial takes polynomial time; we will
show that $\OO((2+\sqrt5)^k)$ independent trials suffice
for constant success probability.

%Let $K$ be the budget after deleting vertices carrying loops. Each trial starts from a fresh copy of this instance and uses the following procedure.

Formally, the core of this algorithm is the procedure $\Solve'(D,k)$ presented below.
Here $\Exact'$ denotes the algorithm of Lemma~\ref{lem:protected}.

\begin{mdframed}[linewidth=.5pt,innertopmargin=8pt,innerbottommargin=8pt]
\noindent\textbf{Procedure $\Solve'(D,k)$.}
\begin{enumerate}[label=\arabic*.]
%\item Set $A:=\emptyset$ and $j:=K$.

\item \textbf{Preprocessing.}
If $D$ is acyclic, return $\emptyset$.
If $k=0$, return $\FAIL$.
Remove from $D$ all vertices/arcs that do not belong to any directed cycle.

\item Compute the weight $\wt(v) = b(v)-1$ for each $v \in V(D)$ and set $L = w(V(D))$.
%their sum $L$, and the set $B$ of nonbimodal vertices.

\item Call $\Exact'(D,P)$ for $P = \{v \in V(D): w(v) > 0\}$.
If $\Exact'(D,P)$ returns a set $Z \subseteq V(D) \backslash P$ with $|Z| \leq k$, then return $Z$.

\item Sample a vertex $v \in V(D)$ with respect to $w$.
    %using Lemma~\ref{lem:approx-sampling} with weights $\wt$ and $N=4k$.
\item Recursively call $\Solve'(D-v,k-1)$.

\item If the recursive call returns $Z \subseteq V(D-v)$, then return $Z\cup\{v\}$.

\item Return $\FAIL$.
\end{enumerate}
\end{mdframed}

%Here $\operatorname{\Complete}$ is the polynomial-time algorithm of Lemma~\ref{lem:exact}, implemented through \textsc{Planar DFAS}. Vertices of $B$ are undeletable only during the call in Step~4; the sets $B$ and $U$ are recomputed after each deletion. We now show that repeating this polynomial-time trial sufficiently many times gives the following result.

\medskip

\begin{theorem}\label{thm:main}
There is a randomized algorithm for \textsc{Planar Directed Feedback Vertex Set}
with running time $(2+\sqrt5)^k n^{\OO(1)}$ and polynomial
space. It always returns either a DFVS of size at most $k$
or $\FAIL$, and returns a solution with probability at least
$\frac{1}{2}$ whenever one exists.
\end{theorem}
\begin{proof}
We analyze the procedure $\Solve'(D,k)$.
As in the proof of Theorem~\ref{thm:baseline}, it suffices to bound its success probability.
Let $\mu(k,p)$ be the (minimum) success probability of a call $\Solve'(D,k)$ satisfying that there exists a DFVS of $D$ with size at most $k$ that contains exactly $p$ positive vertices of $D$.
Clearly, $\mu(k,0) = 1$, because when there is a size-$k$ DFVS of $D$ containing no positive vertices, then Step~3 can successfully find such a DFVS.

Next, we show that $\mu(k,p) \geq \frac{p}{2(k+p)} \cdot \min_{i=0}^{p-1} \mu(k-1,i)$.
It suffices to consider the case where $\Solve'(D,k)$ fails to return a set in Step~3.
Let $P = \{v \in V(D): w(v) > 0\}$ and $X$ be a DFVS of $D$ satisfying $|X| \leq k$ and $|X \cap P| = p$.
The probability that the sampled vertex $v$ lies in $X \cap P$ is equal to $\frac{w(X \cap P)}{L}$.
Note that $w(X \cap P) = w(X)$.
By Lemma~\ref{lem-weight}, $L < 2(|X|+w(X)) \leq 2(k+w(X))$.
Therefore, we have
\begin{equation*}
    \frac{w(X \cap P)}{L} = \frac{w(X)}{L} > \frac{w(X)}{2(k+w(X))} \geq \frac{p}{2(k+p)},
\end{equation*}
where the last inequality follows from the fact $w(X) = w(X \cap P) \geq |X \cap P| \geq p$.
As such, $\Pr[v \in X \cap P] > \frac{p}{2(k+p)}$.
Furthermore, if $v \in X \cap P$, then $\Solve'(D-v,k-1)$ has success probability at least $\min_{i=0}^{p-1} \mu(k-1,i)$, because positive vertices in $D-v$ are also positive in $D$ and hence $X$ contains at most $p-1$ positive vertices in $D-v$.
Based on the above discussion, we see the inequality $\mu(k,p) \geq \frac{p}{2(k+p)} \cdot \min_{i=0}^{p-1} \mu(k-1,i)$ holds.

Define $\mu'(k,p) = \frac{p}{2(k+p)} \cdot \min_{i=1}^{p-0} \mu'(k-1,i)$ for $k \geq p > 0$ and $\mu'(k,0) = 1$ for all $k$.
We have $\mu'(k,p) \leq \mu(k,p)$ for all $k \geq p \geq 0$.
Observe the monotoncity $\mu'(k,p) \leq \mu'(k,p-1) \leq \mu'(k,p-2) \leq \cdots$, which implies $\mu'(k,p) = \frac{p}{2(k+p)} \cdot \mu'(k-1,p-1)$ for all $k \geq p > 0$.
It now suffices to solve this recurrence for $\mu'(k,p)$.
Note that we always have $p \leq k$.
Define $\tau(k,p) = \frac{2(k+p)}{p} \cdot \tau(k-1,p-1)$ and $\tau(k,0)=1$.
By construction, $\mu'(k,p) = \frac{1}{\tau(k,p)}$.
Without loss of generality, assume $k+p$ is even.
Then $\tau(k,p) = 4^p \binom{(k+p)/2}{p} \leq 4^k \binom{(k+p)/2}{p}$.
By the binomial theorem, we have $\binom{a}{b} \leq \frac{\varphi^{2a}}{(4\varphi)^b}$ for $\varphi = 2+\sqrt{5}$.
As such, $\binom{(k+p)/2}{p} \leq \frac{\varphi^{k+p}}{(4\varphi)^p} = \frac{\varphi^k}{4^p}$, which implies $\tau(k,p) \leq \varphi^k$ and $\mu(k,p) \geq \mu'(k,p) \geq  \varphi^{-k}$.

Based on the above discussion, making $\Theta(\varphi^k)$ independent calls of $\Solve'(D,k)$ gives us a success probability at least $\frac{1}{2}$.
Note that $\Solve'(D,k)$ runs in $n^{O(1)}$ time.
So the overall time complexity is $\varphi^k n^{O(1)} = (2+\sqrt5)^k n^{O(1)}$.
\end{proof}

% deterministic_algorithms_clean.tex
% ---------------------------------------------------------------------
% Deterministic algorithms for Planar DFVS.
% This version contains only the two proved bounds:
%   (i)  multilayer iterative compression: 13.279^k n^{\OO(1)}, and
%   (ii) planar-incidence small-face branching: 8.030^k n^{\OO(1)}.
%
% Structural labels used from the preceding sections:
%   lem:bipartite, lem:euler, lem:weight, lem:protected, lem:exact,
%   lem:sampling-monotonicity, sec:bipartite, sec:face-parameters.
% Macros assumed: \wt, \OO, \NestedSolve, \Hbin.
% ---------------------------------------------------------------------

\section{Deterministic algorithms}
\label{sec:deterministic}

\label{sec:deterministic-adaptive}

In this section we give a deterministic algorithm for
\textsc{Planar DFVS}.
Before presenting the algorithm formally, we first discuss its basic intuition.
A natural branching approach to solve the problem is to find a short directed cycle and branch on deleting one of its vertices.
However, even a shortest directed cycle may have arbitrarily
many vertices, so this does not give a bounded number of
branches. Instead, we look for a circular face whose boundary
contains at most $4$ (deletable) non-bimodal vertices, while
allowing arbitrarily many bimodal vertices. We branch on these
non-bimodal vertices and postpone the choice of bimodal solution
vertices. Some vertices become \emph{undeletable} during the
branching in the sense that every solution considered in that branch must avoid
them. We will show that either a suitable circular face exists (so that we can branch on it)
or we can finish by enumerating the non-bimodal part of the
solution and finding the remaining bimodal part in polynomial
time for each guess.

Suppose that an unmarked circular face $f$ has $t \leq 4$ deletable
non-bimodal vertices $v_1,\ldots,v_t$ on its boundary, where
$t\le4$. We create the following branches:
\begin{itemize}
    \item For each $i\in [t]$, delete $v_i$, decrease
    the remaining budget by one, and declare $v_1,\ldots,v_{i-1}$ undeletable.
    \item Declare all of $v_1,\ldots,v_t$ undeletable and \emph{mark} the face $f$.
\end{itemize}
Every solution consistent with the current undeletable vertices is preserved in one of these branches. If it contains a vertex among $v_1,\ldots,v_t$, it is preserved in the branch corresponding
to the first such vertex.
Otherwise, it is preserved in the marking branch. In that branch, the solution must hit the boundary
of $f$ at a deletable bimodal vertex, since $f$ is circular and all its non-bimodal vertices are now undeletable.

Next, we formally describe this algorithm.
The core of the algorithm is the procedure $\textsc{DetSolve}(D,k,U,\mathcal M)$ presented below.
Here $D$ is the graph we are working on, $k$ is the remaining budget, $U$ is the set of undeletable vertices, and $\mathcal{M}$ is a set of marked circular faces.
The initial call is $\textsc{DetSolve}(D,k,\emptyset,\emptyset)$.

\begin{mdframed}[linewidth=.5pt,innertopmargin=8pt,innerbottommargin=8pt]
\noindent\textbf{Procedure}
$\textsc{DetSolve}(D,k,U,\mathcal M)$.
\begin{enumerate}[label=\arabic*.]
\item \textbf{Preprocessing.}
If $D$ is acyclic, return $\emptyset$.
If $k=0$ or $|\mathcal{M}|>2k$, return $\FAIL$.
Remove from $D$ all vertices/arcs that do not belong to any directed cycle.

\item Compute $P = \{v \in V(D) \backslash U: b(v)>1\}$.

\item If $|P| < 3|\mathcal{M}|$, then for each subset $Y \subseteq P$ with $|Y| \leq k - \frac{|\mathcal{M}|}{2}$, we recursively call $\textsc{Exact}'(D-Y,(P \cup U) \backslash Y,k-|Y|)$; if it returns a set $Z \subseteq V(D-Y)$, then return $Y \cup Z$.
If all of these calls return $\FAIL$, then return $\FAIL$.

\item Find a circular face $f \notin \mathcal{M}$ of $D$ that contains at most
$4$ distinct vertices in $P$, and denote these vertices by
$v_1,\ldots,v_t$, where $t\le4$.

\item For each $i\in [t]$, recursively call $\textsc{DetSolve}(D-v_i,k-1,U\cup\{v_1,\ldots,v_{i-1}\},\mathcal{M})$; if it returns a set $Z \subseteq V(D-v_i)$, then return $Z\cup\{v_i\}$.

\item Return $\textsc{DetSolve}(D,k, U\cup\{v_1,\ldots,v_t\},\mathcal{M} \cup \{f\})$.
\end{enumerate}
\end{mdframed}

The following lemma implies that the circular face $f$ in Step~4 always exists (whenever Step~4 is executed), and thus the algorithm is well-defined.

\begin{lemma}\label{lem:small-unmarked-face}
%Let $D$ be a strongly connected plane digraph with at least two vertices, let $B$ be its set of deletable nonbimodal vertices, and let $\mathcal M$ be a set of marked circular faces. Write $q:=|B|$ and $m:=|\mathcal M|$. If $q\ge3m$, then some unmarked circular face has at most four distinct vertices of $B$ on its boundary.
Let $D$ be a strongly connected plane digraph, $\mathcal{M}$ be a set of circular faces of $D$, $U \subseteq V(D)$ such that all vertices incident to the faces in $\mathcal{M}$ are contained in $U$.
Write $P = \{v \in V(D) \backslash U: b(v)>1\}$.
If $|P| \geq 3|\mathcal{M}|$, then there exists a circular face $f \notin \mathcal{M}$ such that $f$ is incident to at most $4$ vertices in $P$.
\end{lemma}
\begin{proof}
Let $\mathcal F$ be the set of circular faces of $D$ that are not in $\mathcal{M}$. By
Lemma~\ref{lem-diface},
\[
    |\mathcal{F}|+|\mathcal{M}|=\ell_{\rm cir}(D)
    \ge 2+\sum_{v\in V(D)}(b(v)-1)
    \ge |P|+2.
\]
The last inequality holds because strong connectivity implies
$b(v)\ge1$ for all $v \in V(D)$ and $b(v)\ge2$ for all $v\in D$.
Thus, if $|P| \geq 3|\mathcal{M}|$, then $|\mathcal{F}|\ge |P|+2-|\mathcal{M}|\ge2$.

Construct the simple bipartite graph $H$ with
parts $P$ and $\mathcal F$: a vertex $v\in P$ is adjacent
to $f\in\mathcal F$ if $v$ is incident to $f$ in $D$.
As a subgraph of the vertex-face incidence graph of $D$, $H$ is clearly planar.
Suppose every face in $\mathcal F$ is incident to at least $5$ vertices in $P$.
Then $|P|\ge5$, and counting edges at the vertices of $\mathcal F$ gives $5|\mathcal{F}|\le |E(H)|\le2(|P|+|\mathcal{F}|)-4$, using the edge bound for simple bipartite planar graphs. Consequently, $3|\mathcal{F}|\le2|P|-4$.
Together with the lower bound on $|\mathcal{F}|$, this implies $3(|P|+2-|\mathcal{M}|)\le2|P|-4$, 
and hence $|P|\le3|\mathcal{M}|-10$, contradicting the fact $|P|\ge 3|\mathcal{M}|$.
Therefore, there exists $f \in \mathcal{F}$ incident to at most $4$ vertices in $P$.
\end{proof}

\begin{theorem}\label{thm:deterministic-8030}
There is a deterministic algorithm for \textsc{Planar Directed Feedback Vertex Set}
with running time $8.03^k n^{\OO(1)}$ and polynomial space.
\end{theorem}
\begin{proof}
We first prove the correctness of our algorithm.
%It suffices to show that if $D$ has a $U$-avoiding DFVS of size at most $k$, then the call $\textsc{DetSolve}(D,k,U,\mathcal{M})$ returns such a DFVS, otherwise it returns \textsc{Fail}.
It is easy to verify that if a call $\textsc{DetSolve}(D,k,U,\mathcal{M})$ returns a set (rather than \textsc{Fail}), then it is a $U$-avoiding DFVS of $D$ with size at most $k$.
So it suffices to show that if $D$ has a DFVS of size at most $k$, our algorithm returns a solution.
%the algorithm returns a subset of $V(D)$, then it must be a DFVS of $D$ with size at most $k$.
Consider the branching tree $T$ of the initial call $\textsc{DetSolve}(D,k,\emptyset,\emptyset)$.
For a node $x \in V(T)$, denote by $D_x,k_x,U_x,\mathcal{M}_x$ the arguments $D,k,U,\mathcal{M}$ of the call corresponding to $x$, respectively, and denote by $P_x$ the set $P$ computed in Step~2 (if the call does not terminate in Step~1).
We say a node $x \in V(T)$ is \textit{good} if there exists a $U_x$-avoiding DFVS $S$ of $D_x$ such that $|S| \leq k_x$ and every vertex in $S \cap P_x$ is not incident to any face in $\mathcal{M}_x$.
%if $V(D) \backslash V(D_x) \subseteq S$, $k \geq |S \cap V(D_x)|$, $U_x \cap S = \emptyset$, and every vertex $v \in S \cap P_x$ is not incident to any face in $\mathcal{M}$.
Clearly, the root of $T$ is good.
We shall show that if $x \in V(T)$ is good, then the call $\textsc{DetSolve}(D_x,k_x,U_x,\mathcal{M}_x)$ returns a solution rather than \textsc{Fail}, which implies the correctness of our algorithm.
The proof is done by induction on the height of $x$ in $T$.
%We observe that if $x$ is a good non-leaf node of $T$, then $x$ has a child in $T$ that is also good.
%Since $x$ is not a leaf, in the call $\textsc{DetSolve}(D_x,k_x,U_x,\mathcal{M}_x)$, Step~5 or Step~6 is executed.
%Let $v_1,\dots,v_t$ be the vertices found in Step~4, which are the vertices in $P$ that are incident to $f$.
%If $\{v_1,\dots,v_t\} \cap S \neq \emptyset$, define $i \in [t]$ as the smallest index such that $v_i \in S$.
%Then the child of $x$ corresponding to the recursive call $\textsc{DetSolve}(D_x-v_i,k_x-1,U_x\cup\{v_1,\ldots,v_{i-1}\},\mathcal{M}_x)$ in Step~5 is good.
%If $v_1,\dots,v_t \notin S$, then the recursive call $\textsc{DetSolve}(D_x,k_x,U_x\cup\{v_1,\ldots,v_t\},\mathcal{M}_x \cup \{f\})$ in Step~6 is good.
%In both cases, $x$ has a good child.
%Next, based on this, we further show that if $x$ is a good node, then $\textsc{DetSolve}(D_x,k_x,U_x,\mathcal{M}_x)$ returns a $U_x$-avoiding DFVS of $D_x$ with size at most $k_x$.
%This can be done by induction on the height of $x$ in $T$.

As the base case, consider a good leaf $x \in V(T)$.
Let $S \subseteq V(D_x)$ be a $U_x$-avoiding DFVS of $D_x$ with size at most $k_x$ such that every vertex in $S \cap P_x$ is not incident to any face in $\mathcal{M}_x$.
If $S = \emptyset$, then $D_x$ is acyclic, then $\textsc{DetSolve}(D_x,k_x,U_x,\mathcal{M}_x)$ returns $\emptyset$ in Step~1, which is a $U_x$-avoiding DFVS of $D_x$ of size at most $k_x$.
Assume $D_x$ is not acyclic.
Then $k_x \geq |S| > 0$.
Also, since the vertices in $S \cap P_x$ are not incident to the faces in $\mathcal{M}$, each face in $\mathcal{M}$ is incident to a vertex in $S \backslash P_x$.
But the vertices in $S \backslash P_x$ are bimodal, and thus each $v \in S \backslash P_x$ can only be incident to at most $2$ faces in $\mathcal{M}_x$.
On the other hand, every face in $\mathcal{M}_x$ is incident to some vertex in $S$, for $S$ is a DFVS.
Therefore, $2|S \backslash P_x| \geq |\mathcal{M}_x|$, which implies $|\mathcal{M}_x| \leq 2k$.
So the call does not terminate in Step~1, and hence terminates in Step~3 since $x$ is a leaf.
In particular, $|P_x| > 3|\mathcal{M}_x|$.
As $2|S \backslash P_x| \geq |\mathcal{M}_x|$, we have $|S \cap P_x| \leq k-|S \backslash P_x| \leq k - \frac{|\mathcal{M}_x|}{2}$.
Thus, when we try $Y = S \cap P_x$ in Step~3, $\textsc{Exact}'(D_x-Y,(P_x \cup U_x) \backslash Y,k_x-|Y|)$ returns a solution, simply because $S \cap (V(D_x) \backslash Y)$ is a $(P_x \cup U_x) \backslash Y$-avoiding DFVS of $D-Y$ and $|S \cap (V(D_x) \backslash Y)| \leq k_x-|Y|$.
As such, the call $\textsc{DetSolve}(D_x,k_x,U_x,\mathcal{M}_x)$ also returns a solution.

Now consider a good non-leaf node $x$, and let $S \subseteq V(D_x)$ be a $U_x$-avoiding DFVS of $D_x$ with size at most $k_x$.
Then in the call $\textsc{DetSolve}(D_x,k_x,U_x,\mathcal{M}_x)$, Step~5 or Step~6 is executed.
It suffices to show one of the recursive calls in Step~5 and Step~6 returns a solution rather than \textsc{Fail}.
Let $v_1,\dots,v_t$ be the vertices found in Step~4, which are the vertices in $P_x$ that are incident to $f$.
Assume all recursive calls return \textsc{Fail}.
Then $x$ has exactly $t+1$ children $y_1,\dots,y_t,y$, where each $y_i$ corresponds to the call $\textsc{DetSolve}(D_x-v_i,k_x-1,U_x\cup\{v_1,\ldots,v_{i-1}\},\mathcal{M}_x)$ in Step~5 and $y$ corresponds to the call $\textsc{DetSolve}(D_x,k_x,U_x\cup\{v_1,\ldots,v_t\},\mathcal{M}_x \cup \{f\})$ in Step~6.
%and $i \in [t]$ be the smallest index such that $v_i \in S$.
If $\{v_1,\dots,v_t\} \cap S \neq \emptyset$, define $i \in [t]$ as the smallest index such that $v_i \in S$.
Then $y_i$ is good (witnessed by the DFVS $S \backslash \{v_i\}$ of $D_x - v_i$), and thus returns a solution by our induction hypothesis.
On the other hand, if $v_1,\dots,v_t \notin S$, then $y$ is good (witnessed by $S$) and returns a solution.
In both cases, there is a contradiction.
As such, at least one recursive call returns a solution, which implies $\textsc{DetSolve}(D_x,k_x,U_x,\mathcal{M}_x)$ also returns a solution.

Next, we analyze the time complexity of our algorithm.
Let $\alpha \in (0,\frac{1}{2})$ be some constant to be determined and $c$ be the number satisfying $\frac{4}{c}+c^{-\alpha} = 1$.
Suppose the choice of $\alpha$ in addition satisfies $(1+\frac{1}{c})^3/\sqrt{c} \leq c^{-\alpha}$.
For each node $x \in V(T)$, define $\mu(x) = k_x - \alpha|\mathcal{M}_x|$.
We shall show that for all $x \in V(T)$, the time cost $\mathsf{Time}(x)$ of the call $\textsc{DetSolve}(D_x,k_x,U_x,\mathcal{M}_x)$ is bounded by $(h(x)+1) c^{\mu(x)} n^{O(1)}$, where $h(x)$ is the height of $x$ in $T$.
%For each node $x \in V(T)$, define $\mu(x) = k_x - \alpha|\mathcal{M}_x|$.
Again, we apply induction on $h(x)$.

As the base case, consider a leaf $x \in V(T)$.
If the call corresponding to $x$ terminates in Step~1, then $\mathsf{Time}(x) = O(n)$.
Otherwise, the call terminates in Step~3, and calls $\textsc{Exact}'$ at most $\sum_{i=1}^{k'} \binom{|P_x|}{i}$ times, where $k' = \min\{|P_x|,k_x-\frac{|\mathcal{M}_x|}{2}\}$.
Note that
\begin{equation*}
    \sum_{i=1}^{k'} \binom{|P_x|}{i} \leq c^{k_x-\frac{|\mathcal{M}_x|}{2}} \left(1+\frac{1}{c}\right)^{|P_x|} \leq c^{k_x} \left(\frac{(1+1/c)^3}{\sqrt{c}}\right)^{|\mathcal{M}_x|} \leq c^{k_x - \alpha|\mathcal{M}_x|} = c^{\mu(x)},
\end{equation*}
where the second inequality follows from the fact $|P_x| < 3|\mathcal{M}_x|$.
So the number of calls of $\textsc{Exact}'$ in Step~3 is bounded by $c^{\mu(x)}$.
Since $\textsc{Exact}'$ runs in polynomial time, we have $\mathsf{Time}(x) = c^{\mu(x)} n^{O(1)} = (h(x)+1) n^{O(1)}$.
%With this definition, if $y$ is a child of $x$, then $\mu(y) = \mu(x)-1$ if $y$ corresponds to a recursive call in Step~5 and $\mu(y) = \mu(x)-\alpha$ if $y$ corresponds to the recursive call in Step~6.

Consider a non-leaf node $x \in V(T)$ and assume $\mathsf{Time}(y) \leq (h(y)+1) c^{\mu(y)} n^{O(1)}$ for all children $y$ of $x$.
Let $y_1,\dots,y_t$ be the children of $x$ corresponding to the recursive calls in Step~5 and $y$ be the child of $x$ corresponding to the recursive call in Step~6.
We have $t \leq 4$.
Note that $\mu(y_i) = \mu(x)-1$ for all $i\in [t]$ and $\mu(y) = \mu(x)-\alpha$.
So we have
\begin{align*}
    \mathsf{Time}(x) &= \sum_{i=1}^t \mathsf{Time}(y_i) + \mathsf{Time}(y) + n^{O(1)} \\
    &\leq (4 h(x) c^{\mu(x)-1} + h(x)c^{\mu(x)-\alpha}+1) \cdot n^{O(1)} \\
    &= ((4/c+c^\alpha) h(x) c^{\mu(x)} + 1) \cdot n^{O(1)} \\
    &= (h(x)c^{\mu(x)} + 1) \cdot n^{O(1)} \\
    &\leq (h(x)+1)c^{\mu(x)} n^{O(1)}.
\end{align*}

We have $h(x) \leq k$ and $\mu(x) \leq k$ for all $x \in V(T)$.
As such, the running time of our algorithm is bounded by $c^{k} n^{O(1)}$.
The choice of $(\alpha,c)$ that minimizes $c$ while satisfying the conditions is (roughly) $\alpha = 0.33$ and $c = 8.03$.
This completes the proof.
\end{proof}
% conclusion_v2.tex -- version 2
% New label: sec:conclusion.
% New citation key to add to the bibliography: CFKLMPPS15 (Cygan et al.,
% Parameterized Algorithms, Springer 2015), for the ETH lower bound for
% Planar Vertex Cover.

\section{Conclusion and open problems}
\label{sec:conclusion}

We gave the first single-exponential fixed-parameter
algorithms for \textsc{Planar DFVS}: a randomized
algorithm running in time
$(2+\sqrt5)^k n^{\OO(1)}$ and a deterministic
algorithm running in time
$8.030^k n^{\OO(1)}$, both using polynomial
space. Both algorithms rely on an exact Euler-type
identity, which implies that every feedback vertex
set of size at most $k$ has weight at least
$L/2+r-k$, together with a polynomial-time algorithm
for finding a minimum feedback vertex set contained
in the deletable bimodal vertices.

Two open problems remain. First, does
\textsc{Planar DFVS} admit a polynomial kernel
that does not increase the solution budget?
More precisely, can an instance $(D,k)$ be reduced
in polynomial time to an equivalent planar instance
$(D',k')$ of size $k^{\OO(1)}$ with $k'\le k$?
Second, does \textsc{DFVS} on general digraphs admit
an algorithm running in time $c^k n^{\OO(1)}$
for some constant $c$?

\paragraph{AI Disclosure}
The authors used Prism (AI Latex Editor, Claude, OpenAI) to generate the current draft based on scaffolding provided by the authors. All proofs are due to the authors without the use of AI tools.

\bibliographystyle{alpha}
\bibliography{references_v4}
\end{document}